\documentclass[oneside,envcountsame,envcountsect,a4paper]{llncs}

\usepackage{paralist} %
\usepackage{amsmath}
\usepackage{amssymb}

\usepackage{graphics}
\usepackage{enumerate}
\usepackage[algoruled,vlined,english]{algorithm2e}
\usepackage{gastex} 
\usepackage{xspace} %

\usepackage{bbding}

\def\ra{\rightarrow}

\newcommand{\macro}[2]{ \providecommand{#1}{{\ensuremath{#2}}\xspace}}

\macro{\N}{\ensuremath{\mathbb N}\xspace}
\macro{\Z}{\ensuremath{\mathbb Z}\xspace}
\macro{\Q}{\ensuremath{\mathbb Q}\xspace}
\macro{\R}{\ensuremath{\mathbb R}\xspace}
\macro{\F}{\ensuremath{\mathbb F}\xspace}

\macro{\bB}{\ensuremath{\mathbf B}\xspace}
\macro{\bC}{\ensuremath{\mathbf C}\xspace}
\macro{\bD}{\ensuremath{\mathbf D}\xspace}
\macro{\bG}{\ensuremath{\mathbf G}\xspace}
\macro{\bH}{\ensuremath{\mathbf H}\xspace}
\macro{\bJ}{\ensuremath{\mathbf J}\xspace}
\macro{\bP}{\ensuremath{\mathbf P}\xspace}
\macro{\bT}{\ensuremath{\mathbf T}\xspace}

\macro{\cA}{\ensuremath{\mathcal A}\xspace}
\macro{\cR}{\ensuremath{\mathcal R}\xspace}
\macro{\cS}{\ensuremath{\mathcal S}\xspace}
\macro{\cF}{\ensuremath{\mathcal F}\xspace}
\macro{\cV}{\ensuremath{\mathcal V}\xspace}
\macro{\cT}{\ensuremath{\mathcal T}\xspace}

\macro{\true}{\textsc{True}}
\macro{\false}{\textsc{False}}

\macro{\LTD}{\textsc{Lt}}
\macro{\wLTD}{\textsc{wLt}}

\macro{\myinput}{{\texttt{in}\xspace}}
\macro{\mem}{{\texttt{mem}\xspace}}
\macro{\xlabel}{{\texttt{x}\xspace}}
\macro{\out}{{\texttt{out}\xspace}}
\macro{\lterm}{{\texttt{term}\xspace}}
\macro{\algo}{\ensuremath{\mathcal A}\xspace}

\newcommand{\NNNN}{\mathcal{P}_{\mathrm{fin}}(\N \times L \times \{0,1\}^* \times \N^2)\xspace}
\newcommand{\Dir}[1]{\ensuremath{\operatorname{Dir}(#1)}\xspace}
\newcommand{\dist}{\ensuremath{\operatorname{dist}}\xspace}

\renewcommand{\epsilon}{\varepsilon}

\macro{\card}{\mathrm Card}

\macro{\coset}{^{\mathrm{c}}}

\macro{\bG}{\mathbf G}
\macro{\bH}{\mathbf H}
\macro{\bK}{\mathbf K}
\macro{\graphs}{\mathcal G}           
\macro{\allg}{\mathcal G}             
\macro{\etiq}{L\times B}                      
\macro{\etiqout}{L_{\texttt{out}}}                      
\macro{\lgraph}{\graphs_\etiq}        
\macro{\ldigraph}{\mathcal{D}_\etiq}        
\macro{\trees}{\mathcal T}            
\macro{\gmin}{\graphs_{\mathrm{min}}} 

\macro{\etiq}{L}                      
\macro{\grs}{\mathrel{{\mathcal R}}}  
\macro{\gfam}{\mathcal F}             
\macro{\covgfam}{\widehat{\gfam}}     

\macro{\mk}{\ensuremath{\mathcal M}}               
\macro{\rtrust}{r^{\mathrm{t}}}       

\macro{\grsdxi}{\grs^\dxi}
\macro{\dxi}{{\chi_\covgfam}}
\macro{\carto}{\mathrm{Carto}}

\macro{\mem}{{\texttt{mem}}}       
\macro{\xlabel}{{\texttt{x}}}      
\macro{\out}{{\texttt{out}}}       
\macro{\lterm}{{\texttt{term}}}    
\macro{\tasks}{\mathfrak T}        
\macro{\task}{T}                   

\macro{\covS}{\widehat S}                 
\macro{\grsi}{\grs^{\mbox{\textsc{i}}}} 
\macro{\grso}{\grs^{\mbox{\textsc{o}}}} 
\macro{\grsg}{\grs^{\mbox{\textsc{g}}}} 
\macro{\choosen}{\mbox{\texttt{choice}}}

\macro{\election}{{\small\textsc{Election}}}
\macro{\elect}{{\small\textsc{Elect}}}
\macro{\nonelect}{{\small\textsc{Non-Elect}}}

\macro{\algo}{\mathcal A}

\begin{document}
\title{Leveraging Structural Knowledge for \\Solving Election in Anonymous Networks\\ with Shared Randomness}

\author{J{\'e}r{\'e}mie
Chalopin \inst{1} \and Emmanuel
Godard \inst{1} }

\institute{
\email{\{jeremie.chalopin,emmanuel.godard\}@lis-lab.fr}\\
CNRS \& Universit{\'e} Aix-Marseille\\
}
\pagestyle{plain}

\maketitle
\begin{abstract}
  We study the classical Election problem in anonymous networks, where
  solutions can rely on the use of random bits, which may be either
  shared or unshared among nodes. We provide a complete
  characterization of the conditions under which a randomized Election
  algorithm exists, for arbitrary structural
  knowledge. Our analysis considers both Las Vegas and Monte Carlo
  randomized algorithms, under the assumptions of shared and unshared
  randomness. In our setting, random sources are considered shared if
  the output bits are identical across specific subsets of nodes.

  The algorithms and impossibility proofs are extensions of those of
  \cite{CGMelection} for the deterministic setting. Our results are a complete generalization of those from \cite{FGL}.
Moreover, as applications, we consider many specific knowledge: no
knowledge, a bound on the size, a bound on the number of nodes sharing a source,
the size, or the full topology of the
network. For each of them, we show how the general characterizations
apply, showing they actually correspond to classes of structural knowledge.
We also describe also how randomized Election algorithms from the literature fits in this landscape.
We therefore provide a comprehensive picture illustrating how
knowledge influences the computability of the Election problem in
arbitrary anonymous graphs with shared randomness.

\end{abstract}

\setcounter{footnote}{0}

\section{Introduction}
The Leader Election, or Election, problem is one of
the paradigms of the theory of distributed computing. 
A distributed algorithm solves the
Election problem if in the final
configuration exactly one process is marked as \textsc{Elected} and all
other processes are labeled \textsc{Non-Elected}.    Election
algorithms constitute a building block for many other distributed
algorithms:  the elected vertex acts as coordinator, initiator, and
more generally performs some special role (cf. \cite[p. 262]{Tanenbaum}). 
The election problem was first studied by LeLann \cite{LeLann} who
gives a solution in rings where each process has a unique name. 
Solutions to this problem are studied under two classical assumptions (see 
\cite[Chapter $3$]{Santoro} for details):
each process is identified by a unique name (or identifier): its identity;
  processes have initially the same state (anonymous networks).
If processes have initially unique identifiers, it is always possible
to solve this problem, e.g., by electing the process with the smallest
identifier. Nevertheless, if we consider
\emph{anonymous}/\emph{homonymous} networks where processes do not
have unique identifiers, it is not
always possible to solve the election problem. Angluin \cite{Angluin}
has introduced the classical proof techniques used for showing the
impossibility of an Election algorithm even knowing the full topology
of the network. This technique is based on graph coverings, which is a
notion known from algebraic topology \cite{Massey}.  Finally, several
characterizations of graphs for which there exists an Election
algorithm have been obtained
\cite{BVelection,YKsolvable,Mazur2,CGMelection}.

It is known that randomized algorithms can solve the Election problem
in situation where Angluin-like results prohibit a deterministic
solution. An early example of such result is \cite{ma89}. It is often
said that randomness helps to perform ``symmetry breaking'' for
Election but it is a misconception to say this is always sufficient to
have a solution. In this paper, we present new impossibility results
that are still valid for randomized algorithms without enough structural
knowledge.
We present two general
characterizations of solvability of Election by a Las Vegas or a Monte
Carlo algorithm given any arbitrary knowledge, in the
general case of shared and unshared random sources.
Our contribution is to specify the exact limit of computability.
In a nutshell, the paper explains how randomness helps, or not, to leverage any
structural knowledge about the underlying network in order to actually
break symmetry, and solve the Election problem.  We also discuss how our impossibility
results match known randomized Election algorithms.

\paragraph*{Computing With Knowledge.}
We encode a given arbitrary knowledge as an arbitrary family of graphs
\gfam, which is the set of graphs with the given value of knowledge (e.g. the
same given size). Solving with this knowledge amounts to find an Election algorithm that is correct on any graph from  
\gfam.
From the general characterizations, we consider, in the Application section, these following structural knowledge that are classically
considered.  We denote \allg the set of all graphs. We consider
\begin{compactitem}
\item \emph{no knowledge} that is $\gfam=\allg$;
\item \emph{bound} $S$ on the size, $\gfam=\{\bG\in\allg, |\bG|\leq S\}$;
\item \emph{strict $2-$approximation} $T$ on the size of the network,
  $\gfam=\{\bG\in\allg, \frac{1}{2}T<|\bG|\leq T\}$;
\item \emph{the size} $S$ of the network, $\gfam=\{\bG\in\allg, |\bG| = S\}$;  
\item \emph{the topology} $\bG$ of the network, $\gfam=\{\bG\}$.
\end{compactitem}
We will see in Section~\ref{sec-app} that this is actually
representative of all the possible computability cases.

It is known from Angluin and later results, that so-called
``symmetric``, or more accurately ``non-minimal'' as described later,
graphs do not admit a deterministic algorithm, even
knowing the exact topology of the network \bG (in this case
$\gfam=\{\bG\}$). 
Here we prove that, given enough structural
knowledge, it is possible to elect with a Monte Carlo randomized
algorithm in any anonymous network. In some sense (that will be made
formal later), for randomized algorithms with unshared sources, all networks are
\textit{minimal}.
This
situation is different from the deterministic case where some
anonymous networks do not admit any Election algorithm, even knowing
the entire topology of the graph.
We therefore
provide a comprehensive picture illustrating how randomness influences
the computability of the Election problem in arbitrary anonymous
graphs.  In the general case, when a bound on the size is known, there
is only a Monte Carlo algorithm. When random sources are
symmetry-breaking, e.g. when there is at least one unshared source,
there is a Las Vegas Election algorithm when the size of the network
is known. When the network is already symmetry-breaking
(covering-minimal) without random sources, then the computability
power of deterministic algorithms is the same as randomized algorithms
for any structural knowledge. Those results demonstrate that
randomness, distributed computability-wise, can help leverage
knowledge to have a randomized Election algorithm when there is no
existing deterministic Election algorithm knowing the topology.
Finally, it is also shown that without enough information on the
network, it could be still impossible to Elect even with a randomized
algorithm, if we require the termination to be explicit.  In some
sense, randoms bits, even shared, do provide symmetry-breaking,
however they are not enough to achieve termination detection.
There is a need for complementary knowledge to leverage the randomization.

\paragraph{Proof Techniques and Contributions.}

A general characterization about leveraging knowledge to solve Election in the
deterministic setting has been done in \cite{CGMelection}.  
The main tool is the notion of quasi-coverings. This notion captures the phenomenon which appears in
the family formed by ``all trees and a triangle'' and quoted by Angluin:
{\it {``the existence of a large enough area of one graph that looks
    like another graph''}} (\cite{Angluin} p. 87, l. 13-17). 
Here we consider networks where nodes have access to (possibly shared) i.i.d. random sources.
In this work, to represent the symmetry relevant to randomized algorithms with
shared sources, we consider graphs where nodes are labeled by the
source of randomness it has access to.  So in the shared randomness
setting, two nodes that share the same source have the same label. The
labeling is denoted by $b$, and the set of nodes having the same
label is called a $B-$class.  However, it is not directly possible to
use the results of \cite{CGMelection} considering $B-$labeled graph
since, in general, a node has no (direct) access to its label %
$b$.

Besides showing that the proof methods from \cite{CGMelection} can be
extended (with some technical care) to the new setting of shared sources,
the main interest of this study are the
applications for standard knowledge, as shown in Section~\ref{sec-app}:
\begin{itemize}
\item there exists just 3 classes of knowledge (exact size/topology, some bounds,
  no knowledge) and we have, through Theorems 1 and 2, a precise
  hierarchy of computability
\item known classical randomized Election Algorithms fit nicely into this
  hierarchy, and our work explains how and why. See also the detailed discussion in Section~\ref{sec:appli}.
\item up to our knowledge, the knowledge considered at Section~\ref{sec-boundedSharing} is the first to consider
  networks of unbounded diameter for which randomized election is still
  solvable. This might seem surprising given the known results from
  the literature.
\end{itemize}

\paragraph*{Related Works.}
Leader Election is a long standing problem in distributed computing.
Angluin introduced graph coverings to capture all the symmetries in
anonymous and homonymous distributed systems.
The notion of symmetry relevant for the message
passing model is captured by \emph{symmetric coverings} of graphs, see Section~\ref{sec:symcov}.
This work has been extended by many researchers and several
characterizations of graphs and deterministic models for which there exists an election
algorithm have been obtained
\cite{BVelection,YKsolvable,Mazur2,CGMelection}.

Shared sources of randomness can be used to improve the complexity
of algorithms as recently shown in \cite{sharedLCL}.
Here, for computability, sharing the same source for
all nodes could actually have an adverse effect. 
Regarding both time and message complexity aspects relative to
knowledge, the most recent advances for randomized Election are given
in \cite{complexityElection21}.

\medskip

The closest previous work using arbitrary knowledge is
\cite{CGMelection} in the deterministic setting, and it is \cite{FGL}
in the general setting of shared randomness.
Relatively to \cite{CGMelection}, the novelty is an extension of the deterministic quasi-covering characterization to the randomized setting.
In \cite{FGL}, the
notion of shared sources of randomness is introduced (under the
terminology ``biased'').  Comparing to \cite{FGL}, we consider the
same randomness setting, but for more general graph topologies than
cliques. It is possible to derive here the characterizations given in
\cite{FGL}, since a clique is $B-$covering minimal if and only if the
gcd of the size of its $B-$classes is 1. %
We underline that our results show the impact of
structural knowledge in the context of randomized algorithms. In some
sense, in cliques, the knowledge of the size, hence the topology of
the network, can be locally derived from the degree. There is no
difference between no knowledge or full knowledge in this family of
graphs.

Note also that we do not consider Election with implicit termination
(i.e. eventually stabilizing to only one \textsc{Elected}
node). Considering knowledge, this setting is solved in \cite{MRZelection},
it is possible to have Election with implicit termination without any
knowledge with a Monte Carlo algorithm.

\section{The Model and Main Technical Results}
Our model is the usual asynchronous message passing model
(\cite{Tel,YKsolvable,CGMelection}). A network is represented by a
simple connected graph $G=(V(G),E(G))$ where vertices correspond to
processes and edges to direct communication links. The initial state
of each process is represented by a label $\lambda(v)$ associated to
the corresponding vertex $v \in V(G)$; we denote by $\bG =
(G,\lambda)$ such a labeled graph. When $\lambda(v)\neq\lambda(u)$ for two different
nodes $u$ and $v$, we can consider that $\lambda(v)$ is the \emph{identity} of the node $v$.
When $\lambda(u)=\lambda(v)$ for all nodes $u,v$, then the network is
said to be \emph{anonymous}. In the general case, the network
is said to be \emph{homonymous}.

We assume that each process can distinguish the different edges that
are incident to it, i.e., for each $u \in V(G)$ there exists a
bijection $\delta_u$ between the neighbors of $u$ in $G$ and
$[1,\deg_G(u)]$. We denote by $\delta$ the set of functions
$\{\delta_u\mid u\in V(G)\}$. The numbers associated by each vertex to
its neighbors are called \emph{port-numbers}, $\delta$ is called a
\emph{port-numbering} of $G$.  Each process $v$ has access at each
step to a source of random bits $b(v)$.
We denote by $b(u,t)$ the $t-$th bit drawn by $u$ at source $b(u)$.
In the case of a shared source between $u$ and $v$, the nodes get always the same output as the other, i.e.
$b(u,t) = b(v,t)$ for all $t$. We also denote $b$ the label corresponding to the random source
that is attached to a given vertex.
We denote by $(\bG,\delta,b)$ the
labeled graph $\bG$ with the port-numbering $\delta$ and sources $b$.

Each process $v$ in the network represents an entity that is capable
of performing computation steps, drawing random bits from its (maybe shared) source,
sending messages via some port and 
receiving any message via some port as was sent by the corresponding
neighbors. %
We consider asynchronous systems, i.e., each computation step
 may take an unpredictable (but finite) amount of
time. We consider only reliable systems: no fault can occur
on processes or communication links. We also assume the channels
are FIFO, i.e., for each channel, the messages are delivered in the
order they have been sent. This delivery depends on a scheduling that is chosen by the adversary.

A deterministic algorithm solves the Election problem on the family $\gfam$ if, for any execution starting from a graph in \gfam, a final configuration is reached where exactly one node is in state \textsc{Elected} and all
the other nodes are in state \textsc{Non-Elected}. We assume that the states \textsc{Elected} and \textsc{Non-Elected} are terminal, i.e., once a node enters in this state, it will not leave this state afterwards. 
A randomized algorithm is an algorithm that can use the random
sources $b$ at any step. We consider a random source to be uniform over time, that is the probability to draw a 1 or a 0 is the same (and non-zero) at any given invocation ; and that the sources are independent and identically distributed.
Given a randomized algorithm and a fixed schedule, the local state of a node $u$ at time $s$ is a random
variable $state(u,s)$, that may depend on the previous state and on the incoming
messages. The probability of an execution is associated to the random variable of the global state. It is therefore the probability of the associated set of random draws at the shared sources, since the schedule is fixed.
The following definitions are standard.
A \emph{Las Vegas algorithm} is an algorithm that terminates with
probability $0<p\leq 1$, and whose final configuration is correct (i.e., there is exactly one node in the \textsc{Elected} state, and all other nodes are in the \textsc{Non-Elected} state).
A \emph{Monte Carlo algorithm} is an algorithm that always terminates,
and whose final configuration is correct with probability $0<p\leq 1$.
A randomized algorithm solves a problem against an asynchronous \emph{non-adaptive scheduler} if it is correct
for any fixed schedule. When the adversary can choose the nodes that are scheduled at the next step knowing not only all local states but also the random bits that were drawn, we say that the adversary is adaptive.
Here, the impossibility proofs
are given in the non-adaptive setting, whereas the correctness of
the algorithms are shown for the adaptive setting.

\subsection{Main Results}

A symmetric covering is a graph homomorphism such that there is a
local bijection at each node, see Section~\ref{sec:symcov} for the
formal definition. From now on, we assume all coverings to be
surjective and when applied to labeled graphs, each label is 
preserved.  A quasi-covering is intuitively a partial graph
homomorphism that behaves like a covering on a subpart of the graph,
see Def.~\ref{def-qc} for a formal definition.  These notions are
extended to the $B-$labeled version of graphs $(\bG,b)$. In order to
clearly distinguish from coverings on the underlying graph, we will
denote the extension $B-$coverings and $B-$quasi-coverings.  A graph
is said to be minimal if any covering is actually an isomorphism.  A
graph $G$ with a $B-$labeling $b$ is $B-$minimal if $(G,b)$ is
minimal.  When there is no pair of nodes sharing a source, all nodes
have a different label and $(G,b)$ is $B-$minimal. This is also the
case when there is at least one unshared source (that is used
by only one node).

We consider any recursive
family of graphs $\gfam$
(encoded as symmetric  $B-$labeled digraphs, see later for the details),
that is, there is an algorithm that decides whether a given graph belongs to $\gfam$.
This is a natural assumption to get actual algorithm parameterized by \gfam, 
consequently, the Election
algorithms we give are the most general election algorithm
possible. It is summarized by the following theorems that consider
both shared and unshared sources through the $B-$labeling.
\begin{theorem}\label{electionLV}
  Let ${\gfam}$ be a recursive family of connected symmetric
  $B-$labeled digraphs.  There exists a Las Vegas Election algorithm
  for ${\gfam}$ if and only if every labeled digraphs of ${\gfam}$ is
  $B-$minimal, and there exists a recursive function $\tau:{\mathcal
    F}\ra \N$ such that for every labeled symmetric digraph $(\mathbf
  D,b)$ of ${\gfam}$, there is no quasi-covering of $\mathbf D$ of radius
  greater than $\tau(\mathbf D)$ in ${\gfam}$, except $\mathbf D$
  itself.
\end{theorem}

Going for Monte Carlo algorithms enables some incorrect runs so it is
possible to consider also non $B-$minimal graphs, and it is needed to
consider only proper quasi-coverings, that is quasi-coverings that are
not actually coverings.

\begin{theorem}\label{electionMC}
  Let ${\gfam}$ be a recursive family of connected symmetric
  $B-$labeled digraphs.  There exists a Monte Carlo Election algorithm
  for ${\gfam}$ if and only if there exists a recursive function
  $\tau:{\mathcal F}\ra \N$ such that for every labeled symmetric
  digraph $\mathbf D$ of ${\gfam}$, there is no proper quasi-covering
  of $\mathbf D$ of radius greater than $\tau(\mathbf D)$ in
  ${\gfam}$, except $\mathbf D$ itself.
\end{theorem}

\section{Symmetric Coverings  and the Election 
Problem for a Labeled Graph} 

\subsection{Preliminaries}\label{sec-prel}
In the following, we will consider directed graphs (digraphs) with
multiple arcs and self-loops. A \emph{digraph} $D=(V(D),A(D),s_D,t_D)$
is defined by a set $V(D)$ of vertices, a set $A(D)$ of arcs and by
two maps $s_D$ and $t_D$ that assign to each arc two elements of
$V(D)$: a source and a target (in general, the subscripts will be
omitted). If $a$ is an arc, the arc $a$ is said to be going out of
$s(a)$ and coming into $t(a)$; we also say that $s(a)$ and $t(a)$ are
incident to $a$. Let $a$ be an arc, if $s(a)= u$ and $t(a)=v$ then 
$v$ is an  out-neighbor of $u$ and $u$ is an in-neighbor of  $v.$
A \emph{symmetric} digraph $D$ is a digraph endowed with a symmetry,
that is, an involution $Sym: A(D) \rightarrow A(D)$ such that for
every $a \in A(D), s(a)=t(Sym(a))$. In a symmetric digraph $D$, the
degree of a vertex $v$ is $\deg_D(v) = |\{a \mid s(a)=v\}| = |\{a \mid
t(a) = v\}|$ and we denote by $N_D(v)$ the set of neighbors of $v$
which is equal to the set of  out-neighbors of $v$ and to the set of 
in-neighbors of $v.$ 

Given two vertices $u, v \in V(D)$, a \emph{path} $\pi$ of {\emph length} $p$
from $u$ to $v$ in $D$ is a sequence of arcs $a_1, a_2, \dots a_p$
such that $s(a_1) = u, \forall i \in [1,p-1], t(a_i)=s(a_{i+1})$ and
$t(a_p)=v$.  If for each $i \in [1,p-1]$, $a_{i+1}\neq Sym(a_i)$,
$\pi$ is \emph{non-stuttering}.  A digraph $D$ is \emph{strongly
  connected} if for all vertices $u, v \in V(D)$, there exists a path
from $u$ to $v$ in $D$.  In a symmetric digraph $D$, the
\emph{distance} between two vertices $u$ and $v$, denoted
$\dist_D(u,v)$ is the length of the shortest path from $u$ to $v$ in
$D$.
In a symmetric digraph $\bD$, we denote by $\bB_\bD(v_0,r)$, the
labeled ball of center $v_0 \in V(D)$ and of radius $r$ that contains
all vertices at distance at most $r$ of $v_0$ and all arcs whose
source or target is at distance at most $r-1$ of $v_0$. 

A \emph{homomorphism} $\gamma$ between the digraph $D$ and the digraph
$D'$ is a mapping $\gamma \colon V(D)\cup A(D) \rightarrow V(D')\cup
A(D')$ such that for each arc $a \in A(D)$, $\gamma(s(a)) =
s(\gamma(a))$ and $\gamma(t(a)) = t(\gamma(a))$.  A homomorphism
$\gamma : D \rightarrow D'$ is an \emph{isomorphism} if $\gamma$ is
bijective, in this case, we note $D \simeq D'$.
Throughout the paper we will consider digraphs where the vertices and
the arcs are labeled with labels from a label set $L\times B$ where
$L$ is a recursive set (representing any additional information
like pseudonyms) and $B$ denotes the set of random sources.
Given a set of labels $L$, a digraph $D$
labeled over $ L$ will be denoted by $(D,\lambda)$, where
$\lambda \colon V(D)\cup A(D) \rightarrow  L$ is the labeling
function.  A mapping $\gamma\colon V(D)\cup A(D) \rightarrow V(D')\cup
A(D')$ is a homomorphism from $(D,\lambda)$ to $(D',\lambda')$ if
$\gamma$ is a digraph homomorphism from $D$ to $D'$ which preserves
the labeling, i.e., such that $\lambda'(\gamma(x))=\lambda(x)$ for
every $x\in V(D)\cup A(D)$.  Labeled digraphs will be designated by
bold letters like $\bD, \bD',\ldots$

With the set of labels $L = \N\times\N$,
we denote by $\ldigraph$ the set of all
symmetric digraphs $\bD = (D,\lambda)$ where for each $a \in A(D)$,
there exist $p, q \in \N$ such that $\lambda(a) = (p,q)$ and
$\lambda(Sym(a)) = (q,p)$ and for each $v \in V(D)$, $\lambda(v) \in
L$ and $\{p \mid \exists a, \lambda(a) = (p,q) \mbox{ and } s(a) = v\}
= [1,\deg_D(v)]$. In other words, $\ldigraph$ is the set of digraphs
that locally look like some digraphs obtained from a simple labeled 
graph $\bG$
with port-numbering labels from $\N$.

To get from graphs to symmetric digraph, we follow the presentation of \cite{CGMelection}.
Let $(G,\lambda)$ be a labeled graph with the port-numbering $\delta.$
We will denote by $(\Dir{\bG},\delta)$ the symmetric
labeled digraph $(\Dir{G},(\lambda,\delta))$ constructed in the
following way.  The vertices of $\Dir{G}$ are the vertices of $G$ and
they have the same labels in $\bG$ and in $\Dir{\bG}$. Each edge
$\{u,v\}$ of $G$ is replaced in $(\Dir{\bG},\delta)$ by two arcs
$a_{(u,v)}, a_{(v,u)} \in A(\Dir{G})$ such that $s(a_{(u,v)}) =
t(a_{(v,u)}) = u$, $t(a_{(u,v)}) = s(a_{(v,u)}) = v$, 
$\delta(a_{(u,v)}) = (\delta_u(v),\delta_v(u))$ and
$\delta(a_{(v,u)}) = (\delta_v(u),\delta_u(v))$. Note that this
digraph does not contain multiple arcs or loops.
The object we use for our study is $(\Dir{G},(\lambda,\delta,b))$ 
and results are stated in full generality for symmetric labeled digraphs.

In the course of an execution, a node is said to be \emph{activable}
if its has pending operation or pending incoming messages. A scheduler
is a family of iterated choices of a non-empty subset of activable
nodes at each step.  A distributed algorithm is probabilistic when
nodes have local access (as formalized above) to a bounded number of
\textit{random bits} at each step. For simplicity, we assume here
there is access to exactly one random bit at each invocation.  
We consider both the asynchronous setting (at each step, an adversary chooses a
subset of activable nodes) and the synchronous setting (at each step, all node are activated and their messages are delivered).  
\subsection{Symmetric Coverings}
\label{sec:symcov}
This section presents a first tool:
symmetric coverings, then it recalls the characterization of labeled
graphs which admit a Las Vegas Election for some knowledge and it
presents the Election algorithm and its main properties.
The algorithm here is an adaptation of \cite{CGMelection}, additional proofs are in the Appendix.
The fundamental notion of symmetric coverings are 
presented in \cite{BVfibrations}. 

A labeled digraph $\bD$ is a \emph{covering} of a labeled digraph $\bD'$
via $\varphi$ if $\varphi$ is a homomorphism from $\bD$ to $\bD'$ such
that each arc $a' \in A(D')$ and for each vertex $v \in
\varphi^{-1}(t(a'))$ (resp. $v \in \varphi^{-1}(s(a'))$, there exists
a unique arc $a \in A(D)$ such that $t(a)=v$ (resp. $s(a) = v$) and
$\varphi(a)=a'$.
A symmetric labeled digraph $\bD$ is a \emph{symmetric covering} of a
symmetric labeled digraph $\bD'$ via $\varphi$ if $\bD$ is a covering
of $\bD'$ via $\varphi$ and if for each arc $a \in A(D)$,
$\varphi(Sym(a)) = Sym(\varphi(a))$. The homomorphism $\varphi$ is a
\emph{symmetric covering projection} from $\bD$ to $\bD'$.
A symmetric labeled digraph $\bD$ is \emph{symmetric-covering
minimal}, or minimal, if there does not exist any symmetric
labeled digraph $\bD'$ not isomorphic to $\bD$ such that $\bD$ is
a symmetric covering of $\bD'$.
To distinguish with $B-$labeled digraphs, we will say that a 
 digraph $(\bD,b)$ is \emph{$B-$symmetric covering
minimal}, or $B-$minimal, when $(\bD,b)$ is minimal. 
Note that when the sources are unshared, each $B-$label occurs only
once, therefore, all networks endowed with at least one unshared source are
$B-$minimal.
We use the following notation, given a set of sources $S$ (or equivalently a set of $B-$labels) and $t\in\N$,
we denote by $X_S$ the random variable corresponding to draws from each sources of $S$. 

\medskip

The following lemma shows the importance of symmetric coverings when
we deal with anonymous networks. This is the counterpart of the
lifting lemma that Angluin gives for coverings of simple graphs
\cite{Angluin} and the proof can be found in
\cite{BVelection,CMasynj}. Here we adapt it to random sources: it is
possible to take the execution steps of the lower graph and carry them
over (aka lifting) to the corresponding nodes, via the covering; in
such a way that the randomized execution still produces on the upper graph $\bD$
a symmetric covering of $\bD'$ with the same probability.

\begin{lemma}[Probabilistic Lifting Lemma]\label{lem-lift}
Let $\bD$ and  $\bD'$ be two $B-$labeled symmetric digraphs of
$\ldigraph.$
If $\bD$ is a symmetric covering of $\bD'$ via $\varphi$, then any
execution of an algorithm $\cA$ on $\bD'$ with probability $p>0$ can be lifted up to an
execution on $\bD$ with probability $p$, such that at the end of the execution, for any $v
\in V(D)$, $v$ is in the same state as $\varphi(v)$.
\end{lemma}
\begin{proof}
  We prove it for one step of execution on $\bD'$, that
  executes with probability $p>0$ on a set of nodes $V'$.  For each $v$ such
  that $\varphi(v) = v'$ with $v'\in V'$, we can apply the same execution step because of the local bijection.
  From the $B-$symmetric covering, all those $v$ share the same random source as
  $v'$, so the new global state of $\bD$ has probability
  $Pr(X_{S'} = x')$, where $S'$ is the set of sources corresponding to $V'$ and $x'$ the bits obtained at this one step execution. This is exactly probability $p$.  
\end{proof}

\subsection{Election in a $B-$Labeled Graph and Symmetric Coverings}
\label{algo-mk}

First, we give a characterization of networks where Election can be
solved in the asynchronous message passing system. This statement
means that for the moment, we consider we know the topology of the graph.
The {\em number of sheets $q$} of a covering is the number of preimages of any node (being a covering, all nodes have the same number of preimages when the covering is surjective).

\begin{theorem}\label{th-main}
Given a $B-$labeled graph $\bG=(G,\lambda,b)$ with a port-numbering $\delta$, 
there exists
a Las Vegas Election algorithm for $(\bG,\delta)$ if and only if
$(Dir(G),(\lambda,{}b,\delta))$ is symmetric covering minimal.
\end{theorem}

\begin{proof}[Necessary part]
The necessary part of this theorem is a direct consequence of Lemma
\ref{lem-lift}. Assume we have a Las Vegas Election algorithm for
$\bG$, consider $\bG'\neq \bG$ such that $\bG$  a covering 
$\bG'$. Even though the Election algorithm is ``meant'' for $\bG$, it is possible
to consider a synchronous execution of this algorithm on $\bG'$. Every step can be
lifted to the synchronous execution on $\bG$, so it will terminate with some positive
probability. Consider now the graphs with final labeling: it is consistent with a
symmetric covering and, since $\bG' \neq \bG$, the number of sheets is greater than 2, there will be
more than one $elected$ node on $\bG$. A contradiction.
\end{proof}

\subsection{A Las Vegas Election Algorithm knowing the Size} 

The sufficient part needs the following naming algorithm (it
will also be used later).
The aim of a naming algorithm is to get to a final configuration
where all nodes have unique identities. Again this is an
essential prerequisite to many other distributed algorithms that work
correctly only under the assumption that all nodes can be
unambiguously identified.  The enumeration problem is a variant of the
naming problem.  The aim of a distributed enumeration algorithm is to
attribute to each network vertex a unique integer in such a way that
this yields a bijection between the set $V(G)$ of vertices and $\{1,
2, \dots, |V(G)|\}$.

In Algorithm~\ref{algo-graph}, we describe a randomized enumeration algorithm knowing the size; by this way  we obtain an election algorithm by considering that the vertex having the number $|V(G)|$ is elected (vertices know $|V(G)|$).
This algorithm is an extension of \cite{CGMelection} which is inspired from the one presented in \cite{CMasynj},
which was an adaptation of the enumeration algorithm given by Mazurkiewicz in 
\cite{Mazur2}.

\textbf{Informal Description.}
We first give a general description of our algorithm $\mk$, when executed on a connected labeled simple graph
$\bG$ with port-numbering $\delta$ and random sources $b$.

During the execution of the algorithm, each vertex $v$ attempts to get
its own unique identity which is a number between $1$ and
$|V(G)|$.
At each step of the algorithm, it invokes the random source $b$ in order to produce incrementally a binary sequence $\overline{b}(v)$. This sequence will eventually be different with the one of other nodes that do not share the same random source.
Once a vertex $v$ has chosen a number $n(v)$, it sends it to
each neighbor $u$ with the initial label, the port-number $\delta_v(u)$ and the sequence $\overline{b}$. When a vertex $u$
receives a message from one neighbor $v$, it stores the number $n(v)$
with the port-numbers $\delta_u(v)$ and $\delta_v(u)$. From all 
information it has gathered from its neighbors, each vertex can
construct its \emph{local view} (which is the set of numbers of its
neighbors associated with the corresponding port-numbers). Then, a
vertex broadcasts its number, its label, its $\overline{b}$ sequence and its mailbox (which contains a set of  {\em local views}).  If a vertex
$u$ discovers the existence of another vertex $v$ with the same tentative number
then it should decide if it changes its own. To this end it
compares its local view with the local view of $v$. If the label of
$u$ together with the random bits sequence $\overline{b}(u)$ or the local view of $u$ is strictly weaker, then $u$ picks another
number --- its new temporary identity --- and broadcasts it again with
its local view and new $\overline{b}$ sequence.
To end of the computation, the node waits until it sees the number
$|V(G)|$ in its mailbox. At this moment if the digraph 
$(Dir(G),(\lambda,\delta))$ is
$B-$symmetric covering minimal, then every vertex will have a unique number from $\{1,
2, \dots, |V(G)|\}$: the algorithm is an Enumeration algorithm. 

\textbf{Labels.}
We consider a network $(\bG,\delta)$ where $\bG = (G,\lambda)$ is a
simple labeled graph and where $\delta$ is a port-numbering of
$\bG$. 
The function $\lambda : V(G) \rightarrow L$ is the
initial labeling.
Note that the $B-$label, that is the random source,
is not known by the node.

We assume there exists a total order $<_L$ on $L.$
We extend the order
$<_L$ to $L \cup \{\bot\}$ (assuming that $\bot \notin L$) as follows:
for all $\ell \in L$, $\bot < \ell $.

We say that $\overline{b}<_B\overline{b'}$ if $\overline{b}$ is before $\overline{b'}$ in the alphabetic order on binary strings.

During the execution, the label of each $v$ is a tuple $(\lambda(v), \overline{b}, n(v), N(v), M(v))$ where:
\begin{compactitem}
\item $\lambda(v) \in L$ is the initial label of $v.$ 
\item $n(v) \in \N$ is the current {\em number} of $v$ computed by the
  algorithm; initially $n(v) = 0$.
\item $\overline{b(v)}$ is a finite sequence of $\{0,1\}$.
\item $N(v) \in \NNNN$, where
  ${\mathcal{P}_{\mathrm{fin}}(-)}$ denotes the set of finite subsets, is the {\em local view} of $v$.
  At the end of the
  execution, if $(m,\ell,\overline{b},p,q) \in N(v)$, then $v$ has a neighbor $u$
  whose number is $m$, whose label is $\ell$, whose random sequence has prefix $\overline{b}$ and the arc from $u$ to
  $v$ is labeled $(p,q)$. Initially $N(v) = \{(0,\bot,\varepsilon,0,q) \mid q \in
  [1,\deg_G(v)]\}$.
\item  
$M(v)$ is a set, it is the {\em mailbox} of $v$;
  initially $M(v) = \emptyset$. 
An element of $M(v)$ has the following form: $(m,\ell,\overline{b},N)$ where
$m\in \N,$ $\ell\in L$, $\overline{b}$ in $\{0,1\}^*$ and $N$ is a local view.
It contains all information received
  by $v$ during the execution of the algorithm. If $(m,\ell,\overline{b},N) \in
  M(v)$, it means that at some previous step of the execution, there
  was a vertex $u$ such that $n(u) = m$, $\lambda(u) = \ell$, $\overline{b}(u)=\overline{b}$ and $N(u)
  = N$.  
\end{compactitem}

\textbf{Messages.}
Nodes exchange messages of the form $ < (n,\ell,\overline{b},
M),p >$. If a vertex $u$ sends a message $ < (n, \ell, \overline{b}, M),p >$
to one of its neighbor $v$, then the message contains the following
information: $n$ is the current number $n(u)$ of $u$, $\ell$ is the
label $\lambda(u)$ of $u$, $\overline{b}$ is the sequence of random bits,
$M$ is the mailbox of $u$, and $p =\delta_{u}(v)$.

\textbf{An Order on Local Views.}
The interesting properties of the algorithm rely on a total order on
local views, which has to be adapted from \cite{CGMelection} to take into account the
sequences of random bits. The idea is to complement the $\lambda$
label with the sequence of random bits, and when comparing such
extended labels, to use the alphabetic order, this way a previous
sequence coming from the same node is always ordered before the
current sequence.

Given two distinct sets $N_1, N_2 \in \NNNN$, we define $N_1 \prec
N_2$ if the maximum of the symmetric difference $N_1 \bigtriangleup
N_2 = (N_1 \setminus N_2) \cup (N_2 \setminus N_1)$ for the
lexicographic order belongs to $N_2$.
One also says that $(\ell,\overline{b},N) \prec (\ell',\overline{b'},N')$ if either $\ell <_L
\ell'$, or $\ell = \ell'$ and $\overline{b}<_B\overline{b'}$; or $\ell = \ell'$ and $\overline{b}=\overline{b'}$ and $N \prec N'$.
We denote by $\preceq$ the reflexive closure of $\prec$.
When doing comparison between mailboxes (equality or
set-inclusion), this is done without considering the bits
sequences, and we keep the same $=$ or $\leq$ operators to have lighter 
notations.

  Consider the mailbox $M = M(v)$ of a vertex $v$ during the execution
  of Algorithm $\mk$ on a graph $(\bG,\overline{b},\delta).$ We say that
  an element $(n,\ell,\overline{b},N) \in M$ is \emph{maximal} in $M$ if there does
  not exist $(n,\ell',\overline{b'},N') \in M$ such that $(\ell,\overline{b},N) \prec
  (\ell',\overline{b'},N')$. We denote by $S(M)$ the set of maximal elements of
  $M$.  From Proposition~\ref{prop-maz}, after each step of
  Algorithm $\mk$, $(n(v),\lambda(v),\overline{b},N(v))$ is maximal in
  $M(v)$.  
  The set $S(M)$ is said \emph{coherent} if it is non-empty and if for
  all $(n_1,\ell_1,\overline{b_1},N_1) \in S(M)$, for all $(n_2,\ell_2,\overline{b_2},p,q) \in N_1$,
  $p \neq 0$, $n_2 \neq 0$ and $\ell_2 \neq \bot$ and for
  $(n_2,\ell_2',\overline{b'},N_2') \in S(M)$ (there is only one by maximality), we have $\ell_2 = \ell_2'$, $\overline{b'}$ is a prefix of $\overline{b_2}$,  and $(n_1,\ell_1,\overline{b_1},q',p')
  \in N_2'$.

\medskip
Action ${\mathbf I}$ can be executed by a node on wake-up
only if it has not received any message. It chooses the
number $1$, updates its mailbox and informs its neighbors.

Action ${\mathbf R}$ describes the instructions the vertex
$v_0$ has to follow when it receives a message 
$<(n_1,\ell_1,\overline{b_1},M_1),p_1>$ from a neighbor via port $q_1$. 
First, it memorizes and it updates its mailbox  by
adding $M_1$ to it. Then it modifies its number if it is equal to $0$ or
if there exists
$(n(v_0),\ell',\overline{b'},{N'}) \in M(v_0)$ such that $(\lambda(v_0),\overline{b}(v_0),N(v_0)) \prec
(\ell',\overline{b'},{N'})$.
The new number is the next available number.
Then, it updates its local view by removing elements which corresponds
to the port $q_1$ (if they exist)
and by adding $(n_1,\ell_1,\overline{b_1},p_1,q_1)$ to $N(v_0)$. Then, it adds its new state
to its mailbox. Finally, if its
mailbox has been modified by the execution of all these instructions,
it sends its number and its mailbox to all its neighbors. 
If the mailbox of $v_0$ is not modified by the execution of the action
${\mathbf R}$, it means that the information $v_0$ has about its neighbor
(i.e., its number) was correct, that all the elements of $M_1$ already
belong to $M(v_0),$ and that for each $(n(v_0),\ell,\overline{b}(v_0),{N}) \in M(v_0)$, $(\ell,\overline{b},{ N})
\preceq(\lambda(v_0),\overline{b}(v_0),N(v_0))$.
This algorithm halts when $n(v_0)=|V(G)|$, and $v_0$ is elected; or when $|V(G)|$ appears in $M(v_0)$ and $v_0$ is non-elected.

Action ${\mathbf C}$ is extending the random bits
sequence in order to help break symmetries between nodes not
sharing the same random source. It is required that $M(v)$ is coherent to apply it so that this does not block progress for the ${\mathbf R}$ rule.

\begin{algorithm}[t]
    \footnotesize
${\mathbf I:}$ \{$n(v_0)=0$ and no message has arrived at $v_0$\}\\
   \Begin{
       $n(v_0):=1$ \;
       $\overline{b}(v_0):=\overline{b}(v_0)rbit()$\;
      $M(v_0):=\{(n(v_0),\lambda(v_0),\overline{b}(v_0),\emptyset)\}$ \;
          \For{$i:=1$ \KwTo $\deg(v_0)$}
          {\KwSty{send} $<(n(v_0),\lambda(v_0),\overline{b}(v_0),M(v_0)),i>$  through
            $i$ \;}}
\BlankLine
${\mathbf R:}$ \{A message $<(n_1,\ell_1,\overline{b_1},M_1),p_1>$ has arrived at
      $v_0$ through port $q_1$\}\\ 
   \Begin{
      $M_{old} := M(v_0)$ \;
      $M(v_0):= M(v_0)\cup M_1$ \;
      \If{$n(v_0) = 0$ or $\exists (n(v_0),\ell',\overline{b'},N') \in M(v_0) \mbox{
      such that } (\lambda(v_0),\overline{b}(v_0),N(v_0)) \prec (\ell',\overline{b'},N')$} 
          {$n(v_0):=1 + \max \{ n' \mid  \exists (n',\ell',\overline{b'},N') \in
            M(v_0)\}$ \;}
      $N(v_0):= N(v_0) \setminus \{(n',\ell',\overline{b'},p',q_1) \mid \exists
          (n',\ell',\overline{b'},p',q_1) \in N(v_0)\} \cup \{(n_1,\ell_1,\overline{b_1},p_1,q_1)
          \}$ \; 
      $M(v_0) := M(v_0)\cup \{(n(v_0),\lambda(v_0),\overline{b}(v_0),N(v_0))\}$ \;
      \If{$M(v_0)\neq M_{old}$} 
         {\For{$i:=1$ \KwTo $\deg(v_0)$}
             {\KwSty{send}
           $<(n(v_0),\lambda(v_0),\overline{b}(v_0),M(v_0)),i>$ through port
           $i$ \;}} 
     }
\BlankLine
${\mathbf C:}$ \{$M(v_0)$ is coherent and there is no $(n,\ell,\overline{b},N) \in M(v_0)$\}\\
\Begin{
    $\overline{b}(v_0):=\overline{b}(v_0)rbit()$\;
    {\For{$i:=1$ \KwTo $\deg(v_0)$}
      {\KwSty{send}
        $<(n(v_0),\lambda(v_0),\overline{b}(v_0),M(v_0)),i>$ through port
        $i$ \;}
    } 
  }
  \caption{Algorithm $\mk$,  with $n =|V(G)|$.\label{algo-graph}}

\end{algorithm}

\subsection{Some Properties of Algorithm $\mk$}

We consider an execution $\rho$ of $\mk$ on $(\bG,b,\delta)$ and for each
vertex $v \in V(G)$, we denote by
$(\lambda(v),n_i(v),\overline{b}_i(v),N_i(v),M_i(v))$ the
state of $v$ after the $i$th computation step of $\rho$ on $v.$
If the vertex $v$ executes an action  
from the step $i$ to the step $i+1,$  it is said active at step $i+1.$
The following proposition summarizes properties that are 
satisfied during an execution
$\rho$ on $(\bG,b,\delta)$. %

\begin{proposition}\label{prop-maz}
Consider a vertex $v$ and a step $i$.
 Then, $n_i(v) \leq n_{i+1}(v)$, $\overline{b}_i(v) \leq_B \overline{b}_{i+1}(v)$,
$N_{i}(v) \preceq N_{i+1}(v)$,  and $M_{i}(v) \subseteq M_{i+1}(v)$.
 For each $(m,\ell,\overline{b},N) \in M_i(v)$ and each $m' \in[1,m]$,
$\exists (m',\ell',\overline{b'},N') \in M_i(v), \exists v' \in V(G)$
such that $n_i(v')=m'$.
\end{proposition}

\begin{proof}
  We suppose that some internal event is executed at step $i+1$ by
  some vertex $v \in V(G)$. The property is obviously true for any
  vertex $w \in V(G) \setminus \{v\}$ and it is easy to see that
  $M_i(v) \subseteq M_{i+1}(v)$.

  If $n_i(v) \neq n_{i+1}(v)$, then $n_{i+1}(v) = 1 + \max \{n' \mid
  (n', \ell',\overline{b}, {\mathcal N'}) \in M_i(v)\}$ and either $n_i(v) = 0 < n_{i+1}(v)$
  or $(n_i(v), \lambda(v), N_i(v)) \in M_i(v)$ 
   and therefore $n_i(v) < n_{i+1}(v)$.

   The random bits sequence is extended by its suffix, so $\overline{b}_{i}(v)\leq_B \overline{b}_{i+1}(v)$.
   
  If $N_i(v) \neq N_{i+1}(v)$, then $v$ has received a message $<
  (n',\cdots,M'), p > $ through port $q$ and $N_{i+1}(v) = N_i(v)
  \setminus \{(n_{old}',p,q)\} \cup \{(n',p,q)\}$ for some (previous) number $n'_{old}$. Let $v'$ be the
  neighbor of $v$ such that $\delta_v(v') = q$ ; we know that
  $\delta_{v'}(v) = p$. 

  If $(n_{old}',p,q) \notin N_i(v)$, then $\max N_{i+1}(v)
  \bigtriangleup N_i(v) = (n',p,q) \in N_{i+1}(v)$ and then $N_i(v)
  \prec N_{+1}(v)$. 

  If $(n_{old}',p,q) \in N_i(v)$, then $n_{old}' \neq n'$. Let $j < i
  + 1$ be the computation step where $v'$ has sent the message $<
  (n',\cdots,M'), p > $. We know that $n_{old}' \leq n' = n_{j}(v')$
  and consequently, $\max N_{i+1}(v) \bigtriangleup N_i(v) = (n',p,q)
  \in N_{i+1}(v)$ and $N_i(v) \prec N_{+1}(v)$. Note that the previous inequality is obtained independently from the $b$ label, by definition of $\prec$.

  \medskip

For the second part of the proposition:
We first note that $(m, \ell, \beta, {\mathcal N})$ is added to $\bigcup\limits_{v \in
  V(G)} M_i(v)$ at some step $i$ only if there exists a vertex $v' \in
  V(G)$ such that $n_i(v') = m$, $\lambda(v') = \ell$, $\overline{b}(v') = \beta$ and $N_i(v') = {\mathcal N}$.

  Given a vertex $v \in V(G)$, a step $i$ and an element $(m,\ell, \beta, {\mathcal N})
  \in M_i(v)$, let $m'\leq M$ and $U = \{(u,j) \in V(G) \times \N \mid j \leq i,
  n_j(u)=m'\}$ and $U' = \{(u,j) \in U \mid \forall (u', j') \in U,
  (\lambda(u'),N_{j'}(u')) \prec (\lambda(u),N_j(u)) \mbox{ or }
  (\lambda(u'), N_{j'}(u')) = (\lambda(u),N_j(u)) \mbox{ and } j' \leq
  j \}$. Since $(m, \ell, \beta, {\mathcal N}) \in M_i(v)$, $U$ and $U'$ are both
  non-empty and it is easy to see that there exists $i_0$ such that
  for each $(u, j) \in U'$, $j = i_0$.
 
  If $i_0 < i$, let $(u, i_0) \in U'$ ; we know that $n_{i_0 + 1}(u)
  \neq n_{i_0}(u)$, but this is impossible, since by maximality of
  $(\lambda(u),\overline{b}(u),N_{i_0}(u))$, $u$ cannot have modified its
  number. Consequently, $i_0 = i$ and there exists $v' \in V(G)$ such
  that $n_i(v') = m'$.
  This ends the proof.
\end{proof}

   From~\cite{CMasynj}, we know that once $n(v),$ $N(v)$ and $M(v)$
   have reached their final values for all $v$, then $S(M(v))$ is
   coherent for any $v$. This is also the case when adding
   $\overline{b}$, thus, if $S(M(v))$ is not coherent, we know that
   $M(v)$ will be modified.

  If the set $S(M)$ is coherent, one can construct a labeled symmetric
  digraph $\bD_M=(D_M,\lambda_M)$ as follows. The set of vertices
  $V(D_M)$ is the set $\{n \mid \exists (n,\ell,N) \in S(M)\}$.  For
  any $(n,\ell,N) \in S(M)$ and any $(n',\ell',p,q) \in N$, there
  exists an arc $a_{n,n',p,q} \in A(D_M)$ such that $t(a) = n, s(a) =
  n'$, $\lambda_M(a) = (p,q)$. Since $S(M)$ is coherent, we can define
  $Sym$ by $Sym(a_{n,n',p,q}) = a_{n',n,q,p}$.

  One can show that Algorithm $\mk$ terminates with probability 1.
  Consider two nodes with different sources, since the sources are
  independent, with probability 1, the finite sequences of random bits
  will eventually be different. So by rule $\mathbf R$, they will have
  different numbers.

  When $\mk$ terminates, 
  the final labeling verifies the following properties: the digraph
  $(Dir(G),(\lambda,\delta))$ is a symmetric covering of $\bD_M$ (see
  Proposition 4.1 in \cite{CMasynj}).  Thus if
  $(Dir(G),(\lambda,\overline{b},\delta))$ is symmetric covering
  minimal then $\bD_M$ is isomorphic to $(Dir(G),(\lambda,\delta))$
  and therefore the set of numbers %
  is exactly
  $[1,|V(G)|]$: each vertex has a unique number. Moreover, 
  termination detection
  of the algorithm
  is possible. %
  Indeed,
  once a vertex gets the identity number $|V(G)|$ (which is known by
  each vertex), from Proposition \ref{prop-maz}, it knows that all the
  vertices have different identity numbers that will not change any
  more and it can conclude that the computation is over. In this case, one can also solve the election problem,
  since this vertex can take the label \emph{elected} and broadcasts
  the information that a vertex has been elected.  Finally, we obtain
  Theorem \ref{th-main} presented above.

\section{Quasi-coverings and The Election Problem for a Family
of Labeled Graphs}

\subsection{Quasi-Coverings}
In the previous section, we assumed the exact size of the network is
known, here we consider general structural knowledge.  This section
presents the second tool we use: quasi-coverings. This tool provides
necessary conditions for the randomized Election in a family of
labeled graphs for both unshared
and shared random sources.

Quasi-coverings have been introduced to study 
the termination detection problem \cite{MMW}. 
The idea behind quasi-coverings is to enable the simulation of local
computations on a given graph  in a restricted area of a larger graph, 
such that
a replay technique can be used to prove impossibility results by
contradiction.  
 The restricted area where we can
perform the simulation will shrink while the number of simulated steps
increases, so the replay technique cannot be used when a bound is known.  
In \cite{GMelection}, the definition of quasi-coverings have been slightly
modified to express more easily this property as a Quasi-Lifting
Lemma.  
The next definition is an
adaptation of this tool to labeled digraphs and is illustrated in Fig.~\ref{fig:qc}.

\begin{definition}\label{def-qc} 
Given two symmetric labeled digraphs $\bD_0, \bD_1$, an integer $r$,
a vertex $v_1 \in V(D_1)$ and a homomorphism $\gamma$ from
$\bB_{\bD_1}(v_1,r)$ to $\bD_0$, the digraph $\bD_1$ is a
\emph{quasi-covering} of $\bD_0$ of center $v_1$ and of radius $r$ via
$\gamma$ if there exists a  symmetric labeled
digraph $\bD_2$ that is a symmetric covering of $\bD_0$ via a
homomorphism $\varphi$ and if there exist $v_2 \in V(D_2)$ and an
isomorphism $\psi$ from $\bB_{\bD_1}(v_1,r)$ to $\bB_{\bD_2}(v_2,r)$
such that for any $x \in V(B_{D_1}(v_1,r)) \cup A(B_{D_1}(v_1,r))$,
$\gamma(x) = \varphi (\psi (x))$.

  We define the {\em number of sheets $q$} to be the
  minimal cardinality of the sets of preimages by $\gamma$: 
 $q = \min_{v\in V(D_0)}|\{w \in\gamma^{-1}(v)| 
             B_{\mathbf D_1}(w,1)\subset B_{\mathbf D_1}(v_1,r)\}|.$
We say that
a quasi-covering is {\em proper} if 
  $B_{\mathbf D_1}(v_1,r-1)$ is not $D_1.$ Any non-proper
quasi-covering is a covering.
\end{definition}

\begin{figure}[t]
\begin{center}
\scalebox{0.6}{\input{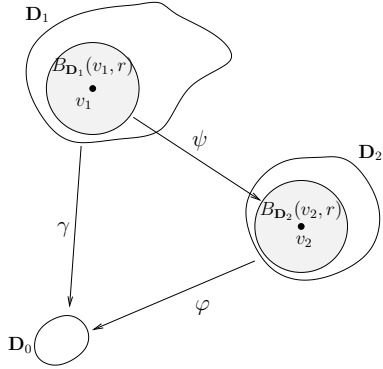}}
  \caption{\label{fig:qcov}Quasi-coverings diagram.
The ball $B_{\bD_1}(v_1,r)$ captures {\it {``the existence of large enough area
of one graph''}} ($\bD_1$) {\it {``that looks locally like another graph''}}
($\bD_0$).\label{fig:qc}}
\end{center}
\end{figure}

\begin{lemma}\label{techlemma}
  Let $\mathbf D_1$ be a proper quasi-covering of $\mathbf D_0$ 
of center $v_1$ and radius $r$ via
  $\gamma$. Then, for any $q\in\N$, if $r \geq q|V(D_0)|$ then $\gamma$
  has at least $q$ sheets. 
\end{lemma}
\begin{proof}
Denote $\mathbf D_2$ the associated covering via $\varphi$. The quasi-covering being
proper, we have that $|B_{\mathbf D_1}(v_1, r)| \geq r \geq q|V
(\mathbf D_0)|$, hence $|V (\mathbf D_2)| \geq q |V (\mathbf D_0)|.$ We
can deduce that $\mathbf D_2$ has at least $q$ sheets.
Now, consider a spanning tree $T$ of $\mathbf D_0$ rooted on $\gamma(v_1)$.
Note $T_1$ the lifting via $\varphi$
of $T$ rooted on $v_2=\psi(v_1)$. By a theorem of Reidemeister \cite{Reidemeister},
there are $q - 1$ disjoint lifted spanning trees
$T_2 , \dots , T_q$ on $\mathbf D_2$ such that the subgraph induced by $T_1 \cup \dots \cup T_q$ is connected.
As $T$ has a diameter at most $|V(\mathbf D_0)| - 1$, we have that
$T_1 \cup \dots \cup T_q \subset B_{\mathbf D_2} (v_2 , q|V (\mathbf D_0)|).$
That means that every vertex of $\mathbf D_0$ has at least $q$ preimages in
$B_{\mathbf D_2} (v_2 , r),$ hence in $B_{\mathbf D_1} (v_1, r).$
\end{proof}

We precise the shrinking of the radius 
after $k$ rounds of a synchronous execution.%

\begin{lemma}[Quasi-Lifting Lemma]\label{quasilifting}
Let  $\bD_1$ be a symmetric 
labeled digraph that is a quasi-covering of $\bD_0$ of
center $v_1$ and of radius $r$ via $\gamma$. Let $k<r$ be a non negative 
integer.  For any algorithm $\cA$, any source labeling $b_0$ and $b_1$ such that
 $(\bD_1,b_1)$ be a symmetric 
labeled digraph that is a quasi-covering of $(\bD_0,b_0)$ of
center $v_1$ and of radius $r$ via $\gamma$.
Let $\bD_0'$ be the digraph obtained after $k$ rounds of a random synchronous 
execution of  $\cA$ on
$(\bD_0,b_0)$ with probability $p>0$. Then  there exists $\bD_1'$ obtained after a random
synchronous execution of $\cA$ on $(\bD_1,b_1)$ with probability $p$ that is a quasi-covering of $\bD_0'$ of center $v_1$ and
of radius $r-k.$
\end{lemma}
\begin{proof}
Consider an algorithm $\cA$ and a digraph $\bD_1 = (D_1, \lambda_1,b_1)$
that is a quasi-covering of $\bD_0 = (D_0, \lambda_0,b_0)$ of center $v_1$
and of radius $r$ via $\gamma$. There exists a symmetric
labeled digraph $\bD_2 =
(D_2,\lambda_2,b_2)$ that is a symmetric covering of $\bD_0$ via a
homomorphism $\varphi$ and a vertex $v_2 \in V(D_2)$ such that
$(B_{D_1}(v,r),\lambda_1,b_1)$ is isomorphic to $(B_{D_2}(v,r),\lambda_2,b_2)$
via an isomorphism $\psi$ and for any $v \in B(v_1,r)$, $\gamma(v)
= \varphi(\psi(v))$.

Let $\bD_0'=(D_0, \lambda_0')$ (resp. $\bD_1'=(D_1, \lambda_1'),
\bD_2'=(D_2, \lambda_2')$) be the labeled digraph where for each
$v$, $\lambda_0'(v)$ (resp.  $\lambda_1'(v)$, $\lambda_2'(v)$) is the
state of $v$ in $D_0$ (resp. $D_1, D_2$) after a computation step of
$\cA$ on $\bD_0$ with probability $p>0$.
To prove the lemma, it is
sufficient to show that we can have a similar step on $\bD_1, \bD_2$ with probability also $p$ while
$\bD_1'$ is a quasi-covering of $\bD_0'$ of
center $v_1$ and of radius $r-1$ via $\gamma$.
From Lemma~\ref{lem-lift}, we know that we can obtain, with probability $p$, $\bD_2'$  that
is a covering of
$\bD_0'$. Moreover, for each $v \in V(B_{\bD_1}(v_1,r-1))$,
$\lambda_1'(v) = \lambda_2'(\varphi(v))$ since
$(B_{\bD_1}(v,1),\lambda_1,b_1)$ is isomorphic to
$(B_{\bD_2}(\varphi(v),1),\lambda_2,b_2)$.
Consequently, $(B_{\bD_1}(v,r-1),\lambda_1',b_1)$ is isomorphic to
$(B_{\bD_2}(v,r-1),\lambda_2',b_2)$ via $\psi$, so $\bD_1'$ is a
quasi-covering of $\bD_0'$ of center $v_1$ and of radius $r-1$ via
$\gamma$.

The probability of such a configuration is $Pr(X_{S_1} = x^1)$ where
$S_1$ is the set of sources in $B_{\bD_1}(v_1,r-1)$ and $x^1$ the bits drawn on these sources.
By composition $\phi\circ\psi$, the set of sources in $\bD_0$ is also $S_1$. Hence this probability is equal to $p$.
\end{proof}

The following is the counterpart of the lifting lemma for
quasi-coverings. 
\vspace{-0.5ex}
\begin{corollary}[Randomized Quasi-Lifting corollary]\label{lem-qlift}
Let $\bD_1$ be a quasi-covering of $\bD_0$ of
center $v_1$ and of radius $r$ via $\gamma$.  For any algorithm $\cA$,
after $r$ rounds of a random synchronous execution of an algorithm $\cA$ on
$\bD_1$ with probability $p>0$, then $v_1$ is in the same state as $\gamma(v_1)$ after $r$
rounds of the synchronous execution of $\cA$ on $\bD_0$ with probability $p$.
\end{corollary}

\subsection{Las Vegas and Monte Carlo Election in a Family of Labeled 
Graphs}%

Using the Randomized quasi-lifting lemma, we give two necessary
conditions for Las Vegas and Monte Carlo Election algorithms.

\begin{proposition}[LV Necessary condition]\label{NC-LV}
  Let $\gfam$ be a recursive 
family of connected  $b-$labeled  
digraphs that are symmetric covering minimal, such
  that there is a Las Vegas Election algorithm  for this family. Then there
  exists a computable function $\tau:{\mathcal I}\ra \N$ such that for
  all labeled digraph $\mathbf D$ of  
  $\gfam$, there is no quasi-covering of $\mathbf D$,
  distinct of $\mathbf D$, of radius greater than 
  $\tau(\mathbf D)$ in  $\gfam$.
\end{proposition}
\begin{proof}
Let  $\cA$ denote an election algorithm  on $\gfam$. 
Consider a
  labeled digraph $\mathbf D\in\gfam$,
  Since this is a Las Vegas algorithm, there exists some terminating
  synchronous execution of $\cA$, that is correct on $\mathbf D\in\gfam$.
  Denote $T$ the number of rounds of this (finite) random synchronous 
  successful execution  on $\mathbf D$.
  We have a corresponding sequence ${\mathcal C}=({\mathbf C_0} = {{\mathbf D}},{{\mathbf
      C}_1},...,{{\mathbf C}_T})$ where
 $\mathbf C_i$ is the  labeled graph obtained after the $i$th
 round. No step of $\cA$ can be applied 
on any vertex of $\mathbf C_T,$ (at the end of the round $t$ no message is 
sent by any vertex).
  This execution is associated with some probability $p>0$.
  Define $\tau(\mathbf D)=
  2|V(\mathbf D)|+ T$.  Then $\tau$ has the desired property.

By contradiction, let $\mathbf D'\in \gfam$ be a quasi-covering
of $\mathbf D$ of radius $\tau(\mathbf D),$ distinct of $\mathbf D.$ By
iteration of Lemma~\ref{quasilifting}, we get with some non null probability $\mathbf D''$ such that
$\mathbf D''$ is a quasi-covering of
$\mathbf C_T$ of radius $\tau(D)-T=2|V(D)|$. The labeled 
digraph $\mathbf D$ 
being symmetric covering minimal and
distinct of $\mathbf D'$, this final quasi-covering $\mathbf D''$ of $\mathbf C_n$ 
is proper. Hence, since by construction  the label ${\it {elected}}$ appears
exactly once in $\mathbf C_T.$, we have that  by Lemma~\ref{techlemma}, the  label ${\it
  {elected}}$ appears at least twice in $\mathbf D''$. A contradiction.
\end{proof}

Now we consider Monte Carlo algorithms. Since they can fail for some executions, the
previous technique does not apply. The radius condition will %
be necessary only for proper quasi-coverings.

\begin{proposition}[MC Necessary condition]\label{NC-MC}
  Let $\gfam$ be a recursive 
family of connected $b-$labeled  
digraphs such
  that there is a Monte Carlo Election algorithm  for this family. Then there
  exists a computable function $\tau:{\mathcal I}\ra \N$ such that for
  all labeled digraph $\mathbf D$ of  
  $\gfam$, there is no proper quasi-covering of $\mathbf D$ of radius greater than 
  $\tau(\mathbf D)$ in  $\gfam$.
\end{proposition}

\begin{proof}
 Let $\cA$ denote a Monte Carlo election algorithm on $\gfam$ with
 correctness $\varepsilon>0$.
We consider the synchronous schedule of algorithm $\cA$ on \bD, Since $\cA$
is a Monte Carlo Election algorithm, there exists a set of random bits
sequences $b_u$, $u\in V(\bD)$, and a time $T\in\N$ at which the
algorithm ends correctly. We denote $p>0$ the probability of the set
of sequences $b_u(1),\cdots, b_u(T)$, $u\in V(\bD)$.

Assume there is a proper quasi-covering $\bK_1$ of \bD of radius $T + 2V(\bD)$.
By
iteration of Lemma~\ref{quasilifting}, we can lift the execution.
We denote $b_1$ the sequence of bits quasi-lifted to $\bK_1$ from $\bD$ on the $T$ rounds.
This sequence has some probability $p_1>0$.
Assume now there is a proper quasi-covering \bK of \bD of radius
$R=N(T+2V(\bD)$, where $N$ is such that $(1-p_1)^N<\frac{1}{2}\varepsilon$. This
proper quasi-covering can be decomposed in $N$ disjoint regions that
define a quasi-covering of radius $T+2V(\bD)$.
For a given quasi-covering, the probability of lifting from \bD and $b$ to $\bK_1$ and $b_1$  
is at least $p_1$. When such a lifting occurs, we have that
the ${\it elected}$ label appears twice on $\bK_1$. Hence, the probability that
there is an incorrect output in one of the $N$ disjoint regions is at
least $1-\frac{1}{2}\varepsilon$,  a contradiction with the level of correctness
for $\cA$.
So there are no proper quasi-covering of $D$ of radius $R$ and 
setting $\tau(\mathbf D)$ as $R$ concludes the proof.
\end{proof}

The sufficient parts needed to prove Thm.~\ref{electionLV} and
\ref{electionMC} are presented in the Appendix since they are
extensions of the algorithms from \cite{CGMelection}.

\section{Applications}

\label{sec:appli}

\subsection{Representative Cases of Structural Knowledge}
We present a summary of the consequences of the two
characterizations and discuss relation with works from the literature.

\label{sec-app}
\begin{figure}[t]
\begin{center}\begin{tabular}{|c|c|c|c|}
\hline & \multicolumn{3}{|c|}{\textbf{$B-$Minimal Graphs}}\\
\hline \textbf{Knowledge} & \textbf{None} & \textbf{Sharing Bound/Bound} & \textbf{2-Approx/Size/Topology}\\
\hline \multicolumn{1}{|l|}{Deterministic} & \raisebox{-3pt}{\color{red}\XSolidBrush}& \raisebox{-3pt}{\color{red}\XSolidBrush}& \raisebox{-3pt}{\color{red}\XSolidBrush}   (\cite{Angluin} for rings)\\
\hline \multicolumn{1}{|l|}{Las Vegas} & \raisebox{-3pt}{\color{red}\XSolidBrush}& \raisebox{-3pt}{\color{red}\XSolidBrush} (\cite{IRelection} for rings)& \raisebox{-3pt}{\color{green}\CheckmarkBold} (\cite{IRelection} for rings) \\
\hline \multicolumn{1}{|l|}{Monte Carlo} & \raisebox{-3pt}{\color{red}\XSolidBrush} (\cite{IRelection} for rings)& \raisebox{-3pt}{\color{green}\CheckmarkBold} (\cite{SSelection} for rings) & \raisebox{-3pt}{\color{green}\CheckmarkBold} (\cite{electionvhp} with vhp)\\
\hline \end{tabular}\end{center}
\caption{Summary of our Election computability results  for $B-$minimal graphs and various knowledge, with previously known results.\label{sumcomp}}
\end{figure}

We consider unshared, shared
randomness, and also bounded sharing randomness (in Section~\ref{sec-boundedSharing}).
There are three cases~: the general case, the $B-$minimal
graphs, and the minimal graphs.  Note that $(\bG,b)$ being $B-$minimal
does not imply that $\bG$ is minimal, but the converse is true.  A
$B-$minimal graph could be an anonymous graph with at least one
unshared source of randomness, a covering-minimal graph or a non
covering-minimal graph where the shared sources do not align with the
symmetry from the covering structure.
Comparing to \cite{FGL}, we remark that $B-$minimality extends the
condition given in \cite{FGL} for cliques. E.g., consider rings, if the
gcd of the size of the $B$ classes is exactly one then the ring is
$B-$minimal.

When a graph is not $B-$minimal, then it does not admit a Las Vegas
Election algorithm even knowing the topology (including the random
sources).  When it is $B-$minimal, then it admits a Las Vegas Election
algorithm provided we have enough knowledge. In particular, a
consequence of Thm.~\ref{electionLV}, is that $B-$minimal graphs
admits a Las Vegas Election algorithm knowing the size.
Knowing a bound on the size does provide a (uniform) $\tau$ function for
limiting proper quasi-covering but not quasi-covering, because
coverings are quasi-coverings of any radius. So for having Las Vegas
Election, enough knowledge to rule out coverings within the same knowledge is necessary. This is the case of 
the strict 2-approximation knowledge, since any strict covering is at least twice as large.
Therefore, while knowing the topology is the strongest possible knowledge, it is actually
not necessary, nor to actually know the size.

What about Monte Carlo Election algorithm in the general case? From
Thm.~\ref{electionMC}, we get that it is not possible to have an
Election algorithm when nothing is known (because this enables
quasi-coverings of arbitrary large radius). However, knowing a bound
on the size does provide a way for limiting proper quasi-covering, but
it is also the case for other bounds, like the sharing bound in
Section~\ref{sec-boundedSharing} that enables a non-uniform $\tau$
function and is therefore a more generic example.  Finally, when
considering minimal graphs, it appears the %
characterizations are equivalent~:

\begin{theorem}
  Let \gfam be a family of covering-minimal graphs, then it is
  possible to solve Election with a deterministic algorithm on \gfam
  if and only if it is possible to solve Election with a Las Vegas algorithm on
  \gfam if and only if it is possible to solve Election with a Monte Carlo
  algorithm on \gfam.
\end{theorem}

We also discuss
previous works, that were only on anonymous rings.
This shows how our general results
are extending all the known cases.
These results are summarized in Fig.~\ref{sumcomp}. 

\begin{compactitem}
\item \cite{Angluin} the seminal work of Angluin is the first proof of the
impossibility of election in anonymous algorithms. it was done in the
local model which is a stronger model than asynchronous message passing.

\item \cite{IRelection-conf,IRelection}: Itai and Rodeh gave an impossibility
  proof for Las Vegas algorithms on anonymous rings knowing a bound.,
  In his textbook
  \cite[Th.9.12]{Tel}, Tel also presents this proof for rings. 
  Itai and Rodeh gave a probabilistic {Election} algorithm when the size
  of the graph is known. Some precise complexity bounds are also given.

\item \cite{SSelection}:Schieber and  Snir present a Monte Carlo algorithm
knowing a bound  on the size.

\item \cite{electionvhp}: Codenotti \emph{et al} present an {Election} algorithm that is
correct with v.h.p\footnote{very high probability} with 1 random bit, knowing the size of the network. 
This result can be easily extended to knowing only a bound using
more random bits.

\end{compactitem}

When the topology is known, probabilistic algorithms are better
solutions than deterministic algorithms. Because it not possible to lift executions forever,
knowing the exact size or topology means it is possible to  wait for Mazurkiewicz' algorithm
to terminate. That there is no enumeration yet can be easily checked 
from the knowledge.

\subsection{Bounded Sharing of Sources of Randomness}
\label{sec-boundedSharing}
We consider now the family $\mathcal B_K$ of graphs where for every
graph, the shared sources are known (that is there is bijection
between the label and the sources) and where the number of nodes
sharing a source with another node is bounded by some number
$K\in\N$. We have a Monte Carlo Election algorithm, this is actually a
corollary of Lemma~\ref{techlemma}.
Since there exists $B-$coverings in $\mathcal B_K$, there is no
Las Vegas algorithm for this knowledge.
We underline that the nodes need to know the class of their random
sources. Otherwise, since networks in $\mathcal B_K$ have unbounded
size, there would be unbounded proper quasi-coverings, so it is
impossible to solve Election in the unlabeled version of $\mathcal B_K$.

\begin{proposition}
  Let $K\geq1$, there is a Monte Carlo Election algorithm for $\mathcal B_K$.
\end{proposition}
\begin{proof}
  Lemma~\ref{techlemma} implies that there is no proper quasi-covering of radius $(K+1)|V(D)|$ for $\bD\in\mathcal B_K$.
  So  $\mathcal B_K$ satisfies the condition for Thm.~\ref{electionMC} with $\tau(\bD) = (K+1)|V(D)|$.
  Therefore Monte Carlo Election is possible in $\mathcal B_K$ for any $K$.
\end{proof}

\appendix

\section{An Election Algorithm for a Family of Labeled
Graphs}\label{sec-algo}

In this section, we present the algorithm we use to obtain our
sufficient conditions corresponding to the previous necessary conditions Prop.~\ref{NC-LV} and Prop.~\ref{NC-MC},
and finally a complete characterization of families of
graphs which admit an Election algorithm as Thm.~\ref{electionLV}
and Thm.~\ref{electionMC} given in the Introduction.

This algorithm is a combination of
the randomized Election algorithm $\mk$ for symmetric minimal labeled digraphs presented
in  Section~\ref{algo-mk}
 and a generalization of an algorithm of Szymanski,
Shy and Prywes (the SSP algorithm for short)~\cite{SSPj}. 
 The SSP algorithm was
originally used to detect the global termination of an algorithm with
local termination provided the nodes initially know a bound on the
diameter of the graph.  

Actually, Algorithm \ref{algo-graph} always
terminates on any network $(\Dir{\bG},\delta)$ and during the
execution, each node $v$ can reconstruct at some computation step
$i$ a symmetric labeled digraph $\bD_i(v)$ such that $(\Dir{\bG},\delta)$ is a
quasi-covering of $\bD_i(v)$. However, this algorithm does not enable
$v$ to compute the radius of this quasi-covering. We use a
generalization of the SSP algorithm to enable each node to compute
a lower bound on the radius of these current quasi-coverings.

\subsection{Termination Detection Mechanism}
A vertex will detect if its state is final by using the bound from the necessary 
condition given by Proposition~\ref{NC-MC}~: it will check whether the computation is stable on a radius greater than the quasi-covering radius bound thus
we add to the label of each vertex two items:
\begin{compactitem}
\item $c(v) \in \Z$ is a counter and initially $c(v) = -1$. In some
  sense, $c(v)$ represents the distance up to which all vertices have
  the same mailbox  as $v$ (up to a suffix on the sequences of random bits). 
\item $A(v) \in {\mathcal{P}_{\mathrm{fin}}(\N\times \N)}$ 
encodes the information $v$ has about the values of
  $c(u)$ for each neighbor $u$. Initially, $A(v)=\{(q,-1)\mid q \in
  [1,\deg_G(v)]\}$.
\end{compactitem}
Thus now
during the execution, the label of each $v$ is a tuple $(\lambda(v),\overline{b}
n(v), N(v), M(v),c(v),A(v)).$

A message sent by a vertex $u$ via the port $p$ to the vertex $v$ has
the following form $ <(n,\ell,\overline{b},M,a),p >$ where $n$ is the current
number $n(u)$ of $u$, $\ell$ is the label $\lambda(u)$, $\overline{b}$ is the current sequence of random bits,
$M$ is the mailbox of $u$, $a$ is the value of the counter $c(u)$ and
$p =\delta_{u}(v)$.

\subsection{ Algorithm $\mk_{\tau}$}
\begin{algorithm}[t]
  \footnotesize
${\mathbf I:}$ \{$n(v_0)=0$ and no message has arrived at $v_0$\}\\
  \Begin{
      $n(v_0):=1$ \;
      $c(v_0):= -1$ \;
      $\overline{b}(v_0):=\overline{b}(v_0)rbit()$\;
      $M(v_0):=\{(n(v_0),\lambda(v_0),\overline{b}(v_0),\emptyset)\}$ \;
          \For{$i:=1$ \KwTo $\deg(v_0)$}
          {\KwSty{send} $<(n(v_0),\lambda(v_0),\overline{b}(v_0),M(v_0)),i>$  through
            $i$ \;}
    }
\BlankLine
${\mathbf R:}$ \{A message $M_1=<(n_1,\ell_1,\overline{b_1},M_1,a_1),p_1>$ has arrived at
      $v_0$ through port $q_1$\}\\ 
   \Begin{
       $M_{old} := M(v_0)$ \;
       $c_{old} := c(v_0)$ \;
       $M(v_0):= M(v_0)\cup M_1$ \;
       \If{ $n(v_0) = 0$ or $\exists (n(v_0),\ell',\overline{b'},N') \in M(v_0)$ 
         such that $(\lambda(v_0),\overline{b}(v_0),N(v_0)) \prec (\ell',\overline{b'},N')$
       }{
         $n(v_0):=1 + \max \{ n' \mid  \exists (n',\ell',N') \in  M(v_0)\}$ \;
         $N(v_0):= N(v_0) \setminus \{(n',\ell',\overline{b'},p',q_1) \mid \exists
          (n',\ell',\overline{b'},p',q_1) \in N(v_0)\} \cup \{(n_1,\ell_1,\overline{b_1},p_1,q_1)
         \}$ \; 
         $M(v_0) := M(v_0)\cup \{(n(v_0),\lambda(v_0),\overline{b}(v_0),N(v_0))\}$ \;
         }
      \If{$M(v_0)\neq M_{old}$}
         {$c(v_0) := -1$ \;
          $A(v_0) := \{(q',-1) \mid 1\leq q'\leq deg(v_0) \}$\;}
      \If{$M(v_0)= M_1$} 
         {$A(v_0) := A(v_0) \setminus \{(q_1,a') \mid \exists (q_1,a') \in
       A(v_0)\} \cup \{(q_1,a_1)\}$ \;} 
      \lIf{$(\forall (q',a') \in A(v_0), c(v_0) \leq a'$ and $\mathbf {QC}(v_0))$}
	 {$c(v_0) := c(v_0) + 1$ \;}
      \If{$M(v_0)\neq M_{old}$ or $c(v_0) \neq c_{old}$} 
         {\For{$i:=1$ \KwTo $\deg(v_0)$}
             {\KwSty{send}
           $<(n(v_0),\lambda(v_0),\overline{b}(v_0),M(v_0),c(v_0)),i>$ through port
               $i$ \;}
         } 
     }
\BlankLine
${\mathbf C:}$ \{$c(v)< \tau(\bD_{M(v)}).$\}\\
\Begin{
    $\overline{b}(v_0):=\overline{b}(v_0)rbit()$\;
    {\For{$i:=1$ \KwTo $\deg(v_0)$}
      {\KwSty{send}
        $<(n(v_0),\lambda(v_0),\overline{b}(v_0),M(v_0)),i>$ through port
        $i$ \;}
    } 
  }

  \caption{Algorithm $\mk_{\tau}.$ ${QC}(v)$ is a shortcut for
  $(S(M(v))\,\text{is coherent})\,\text{and}\,
(c(v)\neq c_{old}(v))\,
\text{and}\, ({\mathbf D}_{M(v)}\in \gfam)\, \text{and}\, 
(c(v)\leq\tau({\mathbf D}_{M(v)})).$
}
\label{algo-general}
\end{algorithm}

We consider a family of graphs that has either the property of
Prop.~\ref{NC-LV} or Prop.~\ref{NC-MC}. So we consider such a function $\tau$ and describe algorithm $\mk_{\tau}$ in Algorithm~\ref{algo-general}. 
We recall that the first rule ${\mathbf
  I}$ can be applied by a node $v$ on wake-up only if it has not
received any message: it takes the number $1$, updates its mailbox and
informs its neighbors. The second rule ${\mathbf R}$ describes the
instructions a node $v$ has to follow when it receives a message
$m$ from a neighbor. It updates its mailbox $M(v)$ and its local view
$N(v)$ according to $m$. Then, if it discovers the existence of
another vertex with the same number and a stronger local view, it
takes a new number greater than all known numbers. 
If its mailbox has not changed, it updates
$A(v)$ and possibly increases $c(v)$ by looking at the counters from the neighbors.
Finally, if $M(v)$ or $c(v)$ has been modified, it informs
its neighbors.

 Using the information stored in its mailbox, each node
$v$ will be able to reconstruct a labeled digraph $\bD$ such that
$(\Dir{\bG},\delta)$ locally looks like $\bD$ up to distance $c(v)$.

\textbf{Properties of  Algorithm $\mk_{\tau}.$}
We consider a graph $\bG$ with a port numbering $\delta$ and an
execution $\rho$ of Algorithm $\mk_{\tau}$ on
$(\bG,\delta)$.  
For each vertex $v \in V(G)$, we note $(\lambda_i(v),\overline{b}_i(v),n_i(v),
N_i(v),M_i(v),$ $ c_i(v), A_i(v))$ the state of $v$ after the $i$th
computation step of $\rho$.

An interesting corollary of Proposition~\ref{prop-maz} is:
 there exists a step $i_0$ such that
after this step for any $v$, the values of $\lambda(v),n(v), N(v)$ and
$M(v)$ are not modified any more.

\begin{proposition}\label{prop-maz-a}
Consider a vertex $v$ and a step $i$.

 If $M_{i}(v) = M_{i+1}(v)$ and if $v$ is active at step $i+1$, then
$c_{i}(v) \leq c_{i+1}(v) \leq c_{i}(v) + 1$ and $c_{i+1}(v) \geq \min\{a
\mid\exists (q,a)\in A_{i+1}(v)\}$ if $\exists (q,a) \in
A_{i+1}(v).$

 If $c_{i+1}(v) \geq 1$,
for each $w \in N_{G}(v)$, there exists a step
$j \leq i$ such that  $c_j(w) \geq
c_{i+1}(v) - 1$ and $M_j(w) = M_{i+1}(v)$.
\end{proposition}

Though the proof uses the same increasing states
argument as \cite{CGMelection}, the proof of Prop.~\ref{prop-maz} is
not a straightforward extension of \cite{CGMelection}. The random bits
sequence is only revealed in a progressive way, the new order on labelings
has to be defined in such a way that 
its dynamic interplay with the
Mazurkiewicz algorithm does not derail the base algorithm.
The key insight is to see (and check carefully) that, using the alphabetical ordering on binary sequences, past labels of a node are always weaker than its current labeling.
The proofs of
Prop.~\ref{prop-maz-a} and Prop.~\ref{prop-qcov-rec} are rather
straightforward extensions of the ones in \cite{CGMelection} since the
generalized SSP algorithm is not a randomized algorithm by itself, so
the bookkeeping is done on $\lambda(v),\overline{b}(v)$ while
considering only $\lambda(v)$ for detecting changes and resetting the
counter $c(v)$. 

\begin{proof}
We prove the proposition by induction on $i.$ Consider a vertex $v$
that modifies its state at step $i+1.$

Initially, if $v$ has applied rule $\mathbf I$ then it is easy to see that the claims hold.

If $v$ has applied rule $\mathbf C$, then the claim holds by induction.

Suppose now that $v$ has applied rule $\mathbf R$ after receiving the 
message $m_1=<(n_1,l_1,b_1,M_1,a_1),p_1>$ via port $q_1$.
Due to the algorithm, we have: $M_i(v)\subseteq M_{i+1}(v).$

If $M_i(v)=M_{i+1}(v)$ and $c_i(v)\neq c_{i+1}(v)$ then $c_{i+1}(v)=c_i(v)+1$.

Let $min_A=min\{a | \exists (q,a)\in A_{i+1}(v)\}.$ 

If $min_A\neq a_1$,
then by induction $c_{i+1}\geq min_A.$ 
If $M_{i+1}(v)\neq M_i(v)$ then $c_{i+1}(v)\geq -1 \geq min_A.$ 
Suppose  that $M_{i+1}(v)=M_i(v).$ If $(q_1,a_1)\in A_i(v)$ then 
$A_{i+1}(v)=A_i(v)$ and by induction $c_{i+1}(v)\geq min_A.$ 
If $M_1\neq M_{i+1}(v)$ then $A_{i+1}(v)=A_i(v)$ and $c_{i+1}(v)\geq min_A.$

Suppose now that
$M_1=M_{i+1}(v).$ If $a_1=0$ then $v$ increases $c(v)$ if it is equal to $-1$ since the $QC$-related inequality is trivially satisfied.
Thus $c_{i+1}(v)\geq 0$ and consequently $c_{i+1}(v)\geq min_A.$

Otherwise, we may assume that $min_A=a_1>0,$ $M_{i+1}(v)=M_i(v)=M_1$
and $(q_1,a_1)\notin A_i(v).$ Let $m_2=<(n_2,\ell_2,\beta_2,M_2,a_2),p_1>$ be the previous
message received via port $q_1.$ Since
communication channels are FIFO, $m_2$ has been sent before $m_1$. 
Since $a_1>0$, $M_2=M_1$ and thus, by induction, $a_2=a_1-1.$ Let $j\leq i$ be
the step where $v$ gets $m_2.$
Since $M_2\subseteq M_j(v)\subseteq M_i(v)=M_1=M_2$,
$M_j(v)=M_2.$ Consequently, $(q,a_2)\in A_i(v).$ Since $a_1=min_A,$
$a_2=a_1-1=min\{a | \exists (q,a)\in A_i(v)\}$. If $c_i(v)\geq a_1$ then
$c_{i+1}(v)\geq c_i(v)\geq min_A.$ Otherwise,
by induction, $c_i(v)=a_2=a_1-1,$ and, since $c_i(v)<a_1$ 
and $c_i(v)\leq min_A,$
$v$ increases $c_i(v).$ Thus $c_{i+1}(v)=1+c_i(v)=a_1\geq min_A.$

Suppose that $c_{i+1}(v)\geq 1$ and consider a vertex $w\in N_G(v).$ From that is above, there 
exists $(\delta_v(w),a)\in A_{i+1}(v)$ with $a\geq c_{i+1}(v)-1\geq 0$ and
$v$ gets  a message $m=<(n,\ell,\beta,M,a),\delta_w(v)>$ from $w$ at a step
$i'\leq i+1.$ Let $j<i+1$ be the step where $w$ sent this message. Since
$a\geq 0$, $M_{i'}(v)=M_{i+1}(v)=M_j(v)$ and 
$c_j(w)=a\geq c_{i+1}(v)-1.$
\end{proof}

\subsection{Quasi-Covering Computation}

We first prove a proposition that enables to present another
definition of quasi-coverings to prove that nodes can compute quasi-coverings of increasing radius.
The proof of this proposition needs the definition of view we give now.

\begin{definition} 
  Consider a symmetric digraph $\bD = (D,\lambda,b) \in \ldigraph$ and a
  vertex $v \in V(D)$.  The \emph{view} of $v$ in $\bD$ is an infinite
  rooted tree denoted by $\bT_{\bD}(v) = (T_D(v),\lambda',b')$ and
  defined as follows:
  \begin{compactitem}
    \item $V(T_D(v))$ is the set of non-stuttering paths $\pi = a_1,
      \dots, a_p$ in $\bD$ with $s(a_1) = v$. For each path $\pi =
      a_1, \dots, a_p$, $\lambda'(\pi) = \lambda(t(a_p)),$ and
      $b'(\pi)=b(t(a_p))$.
    \item for each $\pi, \pi' \in V(T_D(v))$, there are two arcs
      $c_{\pi,\pi'}, c_{\pi',\pi} \in A(T_D(v))$ such that
      $Sym(c_{\pi,\pi'})= c_{\pi',\pi}$ if and only if $\pi' = \pi,
      a$. In this case, $\lambda'(c_{\pi,\pi'}) = \lambda(a)$ and
      $\lambda'(c_{\pi',\pi}) = \lambda(Sym(a))$. Similarly for $b'$ and $b$.
    \item the root of $T_D(v)$ is the vertex corresponding to the
      empty path and its label is $(\lambda(v),b(v))$. 
  \end{compactitem}
\end{definition}

\begin{remark}
For all vertices, $u$ and $v:$ $T_D(u)$ is isomorphic to $T_D(v)$.
We denote by $T_D$ this graph defined up to an isomorphism. It is the
universal covering of $D.$ It is useful to provide examples of
quasi-coverings.
\end{remark}

Consider the view $\bT_\bD(v)$ of a vertex $v$ in a digraph $\bD \in
\ldigraph$ and an arc $a$ such that $s(a) = v$. We define $\bT_{\bD -
  a}(v)$ be the infinite tree obtained from $\bT_\bD(v)$ by removing
the sub-tree rooted in the vertex corresponding to the path $a$.

\begin{proposition}\label{prop-def-qcov}
Given two symmetric labeled digraphs $\bD_0, \bD_1$, an integer $r$,
a vertex $v_1 \in V(D_1)$ and a homomorphism $\gamma$ from
$\bB_{\bD_1}(v_1,r)$ to $\bD_0$, $\bD_1$ is a \emph{quasi-covering} of
$\bD_0$ of center $v_1$ and of radius $r$ via $\gamma$ if and only if
the following holds:

\begin{enumerate}[(i)]
\item for each arc $a \in A(B_{D_1}(v_1,r))$, $\gamma(Sym(a)) =
  Sym(\gamma(a))$,
\item for any $v \in \bB_{\bD_1}(v_1,r)$, $\gamma$ induces an injection
  between the incoming (resp.outgoing) arcs of $v$ and the incoming
  (resp. outgoing) arcs of $\gamma(v)$,
\item for any $v \in \bB_{\bD_1}(v_1,r-1)$, $\gamma$ induces a surjection
  between the incoming (resp.outgoing) arcs of $v$ and the incoming
  (resp. outgoing) arcs of $\gamma(v)$.
\end{enumerate}
\end{proposition}

\begin{proof}    
    If $\bD_1$ is a quasi-covering of $\bD_0$ of center $v_1$ and of
    radius $r$ via $\gamma$, then it is easy to see that $\gamma$
    satisfies these properties.

     Reversely, we construct an infinite covering $\bD_2 =
     (D_2,\lambda_2,b_2)$ of $\bD_0 = (D_0,\lambda_0,b_0)$ as follows. First we
     take a copy $\bB_1$ of $\bB_{\bD_1}(v_1,r)$ and we note $\delta$
     the isomorphism from $\bB_{\bD_1}(v_1,r)$ to $\bB_1$.
     Then, consider a vertex $v$ such that $\dist_{B_1}(v_1,v)=r$ and
     an arc $a_0 \in A(D_0)$ such that $v \in \gamma^{-1}(t(a_0))$. If
     there is no arc $a \in A(B_1)$ such that $t(a) = v$ and $\gamma(a)
     = a_0$, then we add a copy of $\bT_{\bD_0 - a_0}(s(a_0))$ to $\bD_2$ and
     we note $v_2(a_0)$ the root of this tree. We add two arcs $a_2,
     a_2'$ such that $t(a_2) = s(a_2') = v$, $s(a_2) = t(a_2') =
     v_2(a_0)$, $\lambda_2(a_1) = \lambda_0(a_0)$, $\lambda_2(a_2) =
     \lambda_0(Sym(a_0))$ and $Sym(a_2) = a_2'$. Similarly for $b'$ and $b$.

     The digraph $\bD_2$ is the digraph obtained once these
     constructions have been done for all $v \in V(B_1)$ such that
     $\dist_{B_1}(v,v_1)=r$. It is easy to see that $\bD_2$ is a
     symmetric covering of $\bD_0$ via some homomorphism $\varphi$ and
     that for any $v \in V(B_{\bD_1}(v,r))$, $\gamma(v) =
     \varphi(\delta(v))$.
\end{proof}

  The following proposition shows that $\bD_{M(v)}$ has some precise similarity
  with $(\Dir{\bG},\delta)$. 

  \begin{proposition}\label{prop-qcov-rec}
     If $S(M(v))$ is coherent, $(\Dir{\bG},\delta)$ is a
     quasi-covering of $\bD_{M(v)}$ of radius $c(v)$ and center $v$.
  \end{proposition}

The algorithm $\mk_{\tau}$
 is a combination of
the  algorithm $\mk$  presented
in  Section \ref{algo-mk}
 and a generalization of the algorithm of Szymanski,
Shy and Prywes.
The tentative enumeration algorithm will terminate on any network $(\Dir{\bG},\delta).$ with some strictly positive probability.
In the
execution, each node $v$ reconstructs at step
$i$ a symmetric labeled digraph $\bD_i(v)$ such that $(\Dir{\bG},\delta)$ is a
quasi-covering of $\bD_i(v).$ 
When the enumeration algorithm has terminated
then for each vertex $v,$ after the step $t_v$ for the vertex $v$, 
there exists a step $t'_v$ such that  $c(v)\geq \tau(\bD_{M(v)}).$
This knowledge and the knowledge of the reconstructed graph, $\bD_{M(v)},$ 
 implies that each node knows that the enumeration algorithm has terminated
and if its number is the maximal or not among the numbers appearing in the
graph: each node  can decide if it is elected or not.

\begin{proof}%
    Consider a step $i$ and a node $v$ such that $S(M_i(v))$ is
    coherent. If $c_i(v) = 0$, then we are done. Suppose now that
    $c_i(v) \geq 1$. From Proposition~\ref{prop-maz-a}, for each $w \in
    V(Dir(G))$ such that $\dist_{Dir(G)}(v,w) \leq c_i(v)$, there exists a step
    $j_w \leq i$ such that $c_{j_w}(w) \geq c_i(v) - \dist_{Dir(G)}(v,w)$ and
    $M_{j_w}(w) = M_i(v)$.

    Thus, for each $w \in V(B_{{Dir(\bG)}}(v,c_i(v)))$, $j_w$ is defined and
    $c_{j_w}(w) \geq 0$ and for each $w \in V(B_{{Dir(\bG)}}(v,c_i(v) -
    1))$, $c_{j_w}(w) \geq 1$.  For each $w \in
    V(B_{{Dir(\bG)}}(v,c_i(v)))$, $(n_{j_w}(w),\lambda(v),b(v),N_{j_w}(w)) \in
    S(M_{j_w}(w)) = S(M_i(v))$ and thus $\deg_{G}(w) =
    \deg_{D_{M_i(v)}}(n_{j_w}(w))$.  Thus, we can define
    $\gamma(w) = n_{j_w}(w) \in V(D_{M_i(v)})$ and we have
    $\lambda_{M_i(v)}(\gamma(w)) = \lambda(w)$. %

    For each arc $a \in A(B_{{Dir(\bG)}}(v,c_i(v)))$, let $w = t(a), w' =
    s(a)$ (resp. $w=s(a), w'=t(a)$) and suppose without loss of
    generality that $\dist_G(v,w) \leq c_i(v) - 1$. Let $m =
    <(n,\ell,\overline{b},M,a),p>$ be the last message received by $w$ through port
    $\delta_w(w')$ before step $j_w$. Since $c_{j_w}(w) \geq 1$, $M =
    M_{j_w}(w) = M_{j_{w'}}(w')$, $n = n_{j_{w'}}(w')$, $p =
    \delta_{w'}(w)$ and $a \geq 0$.  Thus, we can define
    $\gamma(a) = c_{n,n',p,q}$ (resp. $\gamma(a) = c_{n',n,q,p}$)
    where $n = n_{j_w}(w)$, $n' = n_{j_{w'}}(w')$, $p= \delta_{w'}(w)$
    and $q = \delta_w(w')$. It is easy to see that
    $\lambda_{M_i(v)}(\gamma(a)) = \lambda(a) = (p,q)$
    (resp. $\lambda_{M_i(v)}(\gamma(a)) = \lambda(a) = (q,p)$) and
    that $Sym(\gamma(a))=\gamma(Sym(a))$.
    Consequently, $\gamma$ is a homomorphism from
    $\bB_{{Dir(\bG)}}(v,c_i(v))$ to $\bD_{M_i(v)}$. 

    For each $w \in V(B_{{Dir(\bG)}}(v,c_i(v)))$, for all arcs $a, a'$ such
    that $s(a) = s(a') = w$ (resp. $t(a) = t(a') = w$), $\lambda(a)
    \neq \lambda(a')$ or $\overline{b}(a) \neq \overline{b}(a')$ and thus, $\gamma(a) \neq \gamma(a')$.
    For each $w \in V(B_{{Dir(\bG)}}(v,c_i(v) - 1))$, since $\deg_{G}(w) =
    |N_{j_w}(w)| = \deg_{D_{M_{i}(v)}}(n_{j_w}(w))$, $\gamma$ induces
    a bijection between the incoming (resp. outgoing) arcs of $w$ in
    ${Dir(\bG)}$ and the incoming (resp. outgoing) arcs of $n_{j_w}(w)$ in
    $\bD_{M_i(v)}$.
    From Proposition~\ref{prop-def-qcov}, ${Dir(\bG)}$ is a quasi-covering
    of $\bD_{M_i(v)}$ of center $v$ and of radius $c_i(v)$ via
    $\gamma$.

  \end{proof}

  \begin{remark}
    As can be seen in the proof, the value $\gamma(w)$ does not depend
    on the actual $j_w$. The quasi-covering $\gamma$ is obtained from
    $n$, in the sense that the value of $\gamma(w)$ at $w\in
    B_{Dir(G)}(v,c(v))$ is equal to $n(w)$ at the time-step where
    $M(w)=M(v).$ 
  \end{remark}

\end{document}